\RequirePackage{ifpdf}
\ifpdf % We are running pdfTeX in pdf mode
\documentclass[pdftex]{dsigma}

\else
\documentclass{dsigma}
\fi

\numberwithin{equation}{section}

\numberwithin{theorem}{section}
\numberwithin{proposition}{section}
\numberwithin{lemma}{section}
\numberwithin{corollary}{section}
\numberwithin{definition}{section}
\numberwithin{example}{section}
\numberwithin{remark}{section}
\numberwithin{note}{section}

\begin{document}

\renewcommand{\PaperNumber}{***}

\FirstPageHeading

\ShortArticleName{Integrability of Stochastic Birth-Death processes via Differential Galois Theory}

\ArticleName{Integrability of Stochastic Birth-Death processes via Differential Galois Theory}

% Names of the authors for the title of the paper
\Author{Primitivo B. ACOSTA-HUM\'ANEZ~$^{\dag^1}$$^{\dag^2}$, Jos\'e A. CAPIT\'AN~$^{\dag^3}$$^{\dag^4}$ and Juan J. MORALES-RUIZ~$^{\dag^3}$}

\AuthorNameForHeading{P.B.~Acosta-Hum\'anez, J.A.~Capit\'an and J.J.~Morales-Ruiz}

\Address{$^{\dag^1}$~Instituto Superior de Formaci\'on Docente Salom\'e Ure\~na - ISFODOSU,\\ Recinto Emilio Prud'Homme, Calle R.C. Tolentino \# 51, esquina 16 de Agosto,\\ Los Pepines, Santiago de los Caballeros, Rep\'ublica Dominicana}
\EmailD{\href{mailto:primitivo.acosta-humanez@isfodosu.edu.do}{primitivo.acosta-humanez@isfodosu.edu.do}}
\Address{$^{\dag^2}$~Facultad de Ciencias B\'asicas y Biom\'edicas, Universidad Sim\'on Bol\'ivar,\\Sede 3, Carrera 59 No. 58-135, Barranquilla, Colombia}
% Address of First Author
 % E-mail address of First Author
%\URLaddressD{\url{http://www.home.org/~myHome/}} %URL address of First Author

% Address of Second Author
\Address{$^{\dag^3}$~Depto. de Matem\'atica Aplicada. E.T.S. Edificaci\'on. Avda. Juan de Herrera 6.\\ Universidad Polit\'ecnica de Madrid. 28040, Madrid, Spain}
\EmailD{\href{mailto:ja.capitan@upm.es}{ja.capitan@upm.es}}
\EmailD{\href{mailto:juan.morales-ruiz@upm.es}{juan.morales-ruiz@upm.es}} %
% Address of First Author
\Address{$^{\dag^4}$~Grupo de Sistemas Complejos. Universidad Polit\'ecnica de Madrid. 28040, Madrid, Spain}
% In the case of the same organization, please use the following standard
%\Author{First Names LASTNAME and Second COAUTHOR}
%\AuthoqNameForHeading{F.N. Lastname and S. Coauthor}
%\Address{Address of Author(s), Country}
%\Email{\href{mailto:email@address}{email1@address}, \href{mailto:email@address}{email2@address}}
%\URLaddress{\url{URL1}, \url{URL2})
%Quitar el siguiente comentario para usar el formato de SIGMA
%\ArticleDates{Received ???, in final form ????; Published online ????}

\Abstract{Stochastic birth-death processes are described as continuous-time Markov processes
in models of population dynamics. A system of infinite, coupled ordinary differential equations (the so-called 
master equation) describes the time-dependence of the probability of each system state. Using a generating
function, the master equation can be transformed into a partial differential equation. In this contribution we analyze 
the integrability of two types of stochastic birth-death processes (with polynomial birth and death rates) 
using standard differential Galois theory. We discuss the integrability of the PDE via a Laplace transform 
acting over the temporal variable. We show that the PDE is not integrable except for the (trivial) case in which 
rates are linear functions of the number of individuals.}

%Please type here List of Keywords for your article separated by semicolon.

\Classification{12H05; 35A22; 35C05; 92D25.} % e.g. 35A30; 81Q05
% For 2010 Mathematics Subject Classification see http://www.ams.org/mathscinet/msc/msc2010.html

\section{Introduction}\label{sec:master}

Stochastic birth-death processes~\cite{yule:1924,feller:1939,kendall:1948} are widely used in the mathematical
modeling of interacting populations. They are a special case of continuous-time Markov processes~\cite{karlin:1975}
for which transitions between states are either births (which increase the state variable in one unity) and deaths
(which decrease the state variable by one). Birth-death processes have been used in different fields of applied
science, with many applications in ecology~\cite{nisbet:1982,hubbell:2001,alonso:2006}, queueing theory~\cite{saati:1961}, 
epidemiology~\cite{alonso:2007} and population genetics~\cite{novozhilov:2006}, to mention just a few.
In contrast to deterministic models, these kind of processes make the assumption that population changes take place
in discrete numbers, and this fact introduces variability and noise when compared to deterministic
dynamics~\cite{capitan:2015,capitan:2017}. In the limit of infinite system size, these models are the counterpart
of deterministic dynamics that usually appear in demography and population dynamics~\cite{mckendrick:1912,volterra:1926}.

Only few instances of birth-death processes are analytically tractable in mathematical terms. Most of the
results are related to the probability distributions observed at stationarity~\cite{karlin:1975,haegeman:2011}. Little
is known, however, about how these probability distributions change over time before reaching the equilibrium
state. The existence of closed-form, analytical solutions for certain families of birth-death processes, even when
certain restrictions on the parameters are forced to ensure the existence of analytical solutions, would be a powerful
way to get valuable insights on how probability distributions behave over time and reach the steady state in such
processes. In this contribution we focus on the existence of closed-form analytical solutions for
two widely-used stochastic birth-death models with non-constant birth and death rates.

In order to properly contextualize the problem, we start by describing the mathematical framework used
in the theory of stochastic birth-death processes.
The central quantity used to characterize quantitatively a population formed by a number of individuals is
precisely $N$, the number of individuals observed, where $N\in\mathbb{N}\cup\{0\}$. Let $P_N(t)$ be the
probability that there are $N$ individuals at time $t$ ---the latter variable is regarded as a continuous time---
in the system. Given a particular way of how new individuals enter the system (births events) or leave the
population (death events), the goal of the theory is to describe mathematically this probability. The stochastic
process is fully described once the probability rates $B_N$ (births) and $D_N$ (deaths) are defined.
For simplicity we assume here that rates are time-independent.

As such, birth and death rates, $B_N$ and $D_N$, are regarded as probabilities per unit time that a birth occurs
(hence the populations moves from $N$ to $N+1$ individuals) or, correspondingly, that a death event occurs
(and then the system changes from having $N$ to $N-1$ individuals). Consider an infinitesimal time interval
$\Delta t$. Then birth and death rates satisfy
\begin{equation}
\begin{aligned}
&\text{Pr}\{N+1,t+\Delta t\,|\,N,t\}=B_N\Delta t,\\
&\text{Pr}\{N-1,t+\Delta t\,|\,N,t\}=D_N\Delta t,\\
\end{aligned}
\end{equation}
where $\text{Pr}\{N+1,t+\Delta t\,|\,N,t\}$ is the conditional probability that the system undergoes a birth event
at time $t+\Delta t$ given that there were $N$ individuals at time $t$. Multiple births and deaths are usually
ignored in the limit $\Delta t\to 0$ because their probability would be proportional to $(\Delta t)^2$.

If the population is formed by $N$ individuals at time $t$, at time $t+\Delta t$ the population can be composed
by: (a) $N+1$ individuals with probability $B_N\Delta t$; (b) $N-1$ individuals with probability $D_N\Delta t$;
(c) $N$ individuals with probability $1-\{B_N+D_N\}\Delta t$. Therefore, the conditional probabilities that end up
with a population formed exactly by $N$ individuals are:
\begin{equation}
\begin{aligned}
&\text{Pr}\{N,t+\Delta t\,|\,N-1,t\}=B_{N-1}\Delta t,\\
&\text{Pr}\{N,t+\Delta t\,|\,N+1,t\}=D_{N+1}\Delta t,\\
&\text{Pr}\{N,t+\Delta t\,|\,N,t\}=1-(B_N+D_N)\Delta t.
\end{aligned}
\end{equation}
Thus, using the theorem of total probability we can write an expression for the probability of observing $N$
individuals at time $t+\Delta t$ in terms of the probabilities at time $t$:
\begin{equation}\label{eq:total}
P_N(t+\Delta t)=P_{N-1}(t)B_{N-1}\Delta t+P_{N+1}(t)D_{N+1}\Delta t+P_N(t)[1-(B_N+D_N)\Delta t].
\end{equation}
Here we are assuming a Markovian hypothesis, according to which the state of the system at a given time
is conditioned only by the potential states at previous times but infinitely close to the current time. Now we
subtract $P_N(t)$ from both sides of~\eqref{eq:total} and take the limit $\Delta t\to 0$ to get the so-called
\emph{master equation}:
\begin{equation}\label{eq:master}
P'_N(t)=B_{N-1}P_{N-1}(t)+D_{N+1}P_{N+1}(t)-(B_N+D_N)P_N(t).
\end{equation}
The system is therefore fully described by a coupled system of infinitely many ordinary differential equations, given by
equation~\eqref{eq:master} for $N\ge 1$. For $N=0$, since the number of individuals has to remain non-negative,
we have to impose that $D_0=0$ and $B_{-1}=0$. In this case, the corresponding equation reduces to
\begin{equation}\label{eq:master1}
P'_0(t)=D_1P_1(t)-B_0P_0(t).
\end{equation}
Therefore, the central problem in the theory of stochastic birth-death processes for a population of individuals
is to solve the master equation~\eqref{eq:master}--\eqref{eq:master1} as a problem of initial value. To be precise,
if we know the probability distribution $P_N(t_0)$ at some initial time $t_0$, then the system of differential equations
allows for obtaining the probability distribution at any time $t$, $P_N(t)$. In what follows we will assume, without loss
of generality, that there are $N_0$ individuals at time $t=0$, $N_0\in\mathbb{N}\cup\{0\}$. Then the initial probability
distribution is precisely equal to $P_N(0)=\delta_{N_0N}$, where $\delta_{ij}$ stands for the usual Kronecker delta
symbol.

To make the master equation tractable, in some cases it can be transformed into a partial differential equation
by using a generating function defined as
\begin{equation}\label{eq:genfun}
g(z,t)=\sum_{N=0}^{\infty}P_N(t)z^N,
\end{equation}
i.e., the discrete variable $N$ is transformed into to continuous variable $z$, $0\le z\le 1$. In this contribution we are
interested in the conditions under which we can find a closed-form, analytical solution for the generating function
$g(z,t)$. Knowledge of the generating function allows the calculation of important properties of the stochastic
processes ---for example, the average number of individuals, the variance of the population, or even the probability of extinction
of the system at time $t$, $P_0(t)=g(0,t)$. We will  follow a Laplace transform strategy to solve the corresponding partial
differential equation, and we will analyze ecologically meaningful examples for the birth and death rates, which yield
useful insights about the integrability of these kind of systems by considering one fully integrable case and a fully
non-integrable case (the sense in which we use the term `integrable' will be precisely defined in Section~\ref{sec:DGT}).

To be more specific, from now on we focus on the birth-death process defined by the rates (as mentioned, the term `rate' stands
for probability per unit time) $B_N=\beta N^{b}$ and $D_N=\delta N^{d}$, where $b$, $d$ are natural exponents 
and $\beta$, $\delta$ are positive real numbers. Usually death rates are taken as a quadratic function ($d=2$) since it is commonly
assumed that two individuals compete with each other in death events, whereas birth processes (asexual reproduction) are described as
linear functions of $N$ ($b=1$, i.e., the probability of a birth event is proportional to the number of individuals in
the population). In this contribution we will consider two combinations of exponents: $(b,d)\in\{(1,1),(1,2)\}$. For
the combination $(b,d)=(1,1)$, the master equation is equivalent to the following PDE (see Section~\ref{sec:11}):
\begin{equation}\label{eq:PDE11}
\frac{\partial g(z,t)}{\partial t}=(1-z)(\delta -\beta z)\frac{\partial g(z,t)}{\partial z}
\end{equation}
with boundary conditions
\begin{equation}\label{eq:boundPDE}
g(z,0)=z^{N_0}, \,\,g(1,t)=1.
\end{equation} This equation turns out to be integrable (in a sense
specified below) via a Laplace transform technique.
Roughly speaking, integrability here means solvability in  {\it  closed form}.

If death rates are quadratic functions of $N$, in Section~\ref{sec:12} we show that the generating function satisfies the following PDE,
\begin{equation}\label{eq:PDE12}
\frac{\partial g(z,t)}{\partial t}=(1-z)\left[(\delta -\beta z)\frac{\partial g(z,t)}{\partial z}+\delta z \frac{\partial^2 g(z,t)}{\partial z^2}\right].
\end{equation}
We will refer to this PDE as the $(b,d)=(1,2)$ case. As before, the generating function has to satisfy the conditions \eqref{eq:boundPDE}.
This case of quadratic death rates, which is the more relevant one in biological
terms, remains as non-integrable, as we will show in Section~\ref{sec:12} using results from Differential Galois Theory.

\begin{proposition} \label{prop 11} The  PDE given by equation~\eqref{eq:PDE11} with  boundary conditions \eqref{eq:boundPDE} is integrable and the solution is given by:
\begin{enumerate}
\item For $\beta\neq\delta$,
\begin{equation}\label{eq:genfun11}
g(z,t)=\left[\frac{\delta-\beta z-(1-z)\delta e^{(\beta-\delta)t}}{\delta-\beta z-(1-z)\beta e^{(\beta-\delta)t}}\right]^{N_0}.
\end{equation}
\item For $\beta=\delta$,
\begin{equation}\label{genebeq}
 g(z,t)=\left[\frac {\delta(1-z)t+z}{\delta(1-z)t+1}\right]^{N_0}.
\end{equation}
\end{enumerate}
\end{proposition}
\begin{proposition} \label{prop 12} The  PDE given by equation~\eqref{eq:PDE12} is non-integrable.\end{proposition}
Proposition~\ref{prop 12} is a non-integrability result and  tell us that {\it  any search for a closed-form, analytical solution for equation~\eqref{eq:PDE12} is doomed to failure.}

\section{Differential Galois Theory}
\label{sec:DGT}

Differential Galois Theory, also known as Picard-Vessiot theory, is the Galois theory of linear differential
equations. In classical Galois theory, the main object is a
group of permutations of the polynomial's roots, whereas in the Picard-Vessiot
theory it is a linear algebraic group. For polynomial equations we
look for solutions in terms of radicals. According to classical Galois
theory, this form of the solution will exist whenever the Galois group is a solvable group.
An analogous situation holds for linear homogeneous differential
equations.

As a notational convention we will use $\partial_x:=\frac{\partial}{\partial x}$ (also $^{\prime}:=\frac{\partial}{\partial x})$ throughout this section.

\subsection{Definitions and Known Results}
The following theoretical background can be found in the references \cite{crha,mo,vasi}. We recall that although differential Galois theory is more 
general, here we just summarize results from theory for second order differential equations.

\begin{definition}[Differential Fields]\label{defdiff}
Let $K$ (depending on $x$) be a commutative field of
characteristic zero, and $\partial_x$ a derivation, that is, a map
$\partial_x : K\rightarrow K$ satisfying $\partial_x (a + b) =
\partial_x a + \partial_x b$ and $\partial_x(ab) = \partial_x a \cdot b + a
\cdot\partial_x b$ for all $a,b\in K$. By $\mathcal C$ we denote
the field of constants of $K$, $$\mathcal C = \{c\in K\, |\,\,
c' = 0\},$$ which is also of characteristic zero and will be assumed algebraically closed. In this terms,
we say that $K$ is a {\it{differential field}} with the derivation
$\partial_x=\,\,^\prime$.
\end{definition}

Up to special considerations, we
analyze second order linear homogeneous differential equations,
that is, equations in the form
\begin{equation}\label{soldeq}
\mathcal L:= y''+ay'+by=0,\quad a,b\in K.
\end{equation}

\begin{definition}[Picard-Vessiot Extension] Suppose that $y_1, y_2$ is a basis of solutions of $\mathcal L$ given in equation \eqref{soldeq}, i.e., $y_1, y_2$ are linearly
independent over $K$ and every solution is a linear combination
over $\mathcal{C}$ of these two. Let $L= K\langle y_1, y_2 \rangle
=K(y_1, y_2, y_1', y_2')$ the differential
extension of $K$ such that $\mathcal C$ is the field of constants
for $K$ and $L$. In this terms, we say that $L$, the smallest
differential field containing $K$ and $\{y_{1},y_{2}\}$, is the
\textit{Picard-Vessiot extension} of $K$ for $\mathcal L$.
\end{definition}

\begin{definition}[Differential Galois Groups]
Assume $K$, $L$ and $\mathcal L$ as in the previous definition. The
group of all differential automorphisms (automorphisms that
commute with derivation) of $L$ over $K$ is called the
{\it{differential Galois group}} of $L$ over $K$ and is denoted by
${\rm Gal}(L/K)$. This means that for $\sigma\in
\mathrm{Gal}(L/K)$, $\sigma(a)=(\sigma(a))'$
for all $a\in L$ and for all $a\in K,$ $\sigma(a)=a$.
\end{definition}
Assume that $\{y_1,y_2\}$ is a fundamental system (basis) of solutions
of $\mathcal L$. If $\sigma \in
\mathrm{Gal}(L/K)$ then $\{\sigma y_1, \sigma y_2\}$ is another
fundamental system of $\mathcal L$. Hence there exists a matrix

$$A_\sigma=
\begin{pmatrix}
a & b\\
c & d
\end{pmatrix}
\in \mathrm{GL}(2,\mathbb{C}),$$ such that
$$\sigma
\begin{pmatrix}
y_{1}\\
y_{2}
\end{pmatrix}
=
\begin{pmatrix}
\sigma (y_{1})\\
\sigma (y_{2})
\end{pmatrix}
=\begin{pmatrix} y_{1}& y_{2}
\end{pmatrix}A_\sigma.$$ 
In a natural way, we can extend this to systems:
$$\sigma
\begin{pmatrix}
y_{1}&y_2\\
y_1'&y_{2}'
\end{pmatrix}
=
\begin{pmatrix}
\sigma (y_{1})&\sigma (y_2)\\
\sigma (y_1')&\sigma (y_{2}')
\end{pmatrix}
=\begin{pmatrix} y_{1}& y_{2}\\y_1'&y_2'
\end{pmatrix}A_\sigma.$$

This defines a faithful representation $\mathrm{Gal}(L/K)\to
\mathrm{GL}(2,\mathbb{C})$ and it is possible to consider
$\mathrm{Gal}(L/K)$ as a subgroup of $\mathrm{GL}(2,\mathbb{C})$.
It depends on the choice of the fundamental system $\{y_1,y_2\}$,
but only up to conjugacy.

One of the fundamental results of the Picard-Vessiot theory is the
following theorem (see \cite{ka,kol}).

\begin{theorem}  The differential Galois group $\mathrm{Gal}(L/K)$ is an
algebraic subgroup of $\mathrm{GL}(2,\mathbb{C})$.
\end{theorem}

\begin{definition}[Integrability]\label{integrability} Consider the linear differential equation
$\mathcal L$ such as in equation \eqref{soldeq}. We say that
$\mathcal L$ is \textit{integrable} if the Picard-Vessiot
extension $L\supset K$ is obtained as a tower of differential
fields $K=L_0\subset L_1\subset\cdots\subset L_m=L$ such that
$L_i=L_{i-1}(\eta)$ for $i=1,\ldots,m$, where either
\begin{enumerate}
\item $\eta$ is {\emph{algebraic}} over $L_{i-1}$, that is $\eta$ satisfies a
polynomial equation with coefficients in $L_{i-1}$.
\item $\eta$ is {\emph{primitive}} over $L_{i-1}$, that is $\partial_x\eta \in L_{i-1}$.
\item $\eta$ is {\emph{exponential}} over $L_{i-1}$, that is $\partial_x\eta /\eta \in L_{i-1}$.
\end{enumerate}
\end{definition}

We remark that the usual terminology in differential algebra for
integrable equations is that the corresponding Picard-Vessiot
extensions are called \textit{Liouvillian}.

\begin{theorem}[Kolchin]\label{LK}
The equation $\mathcal L$ given in \eqref{soldeq} is integrable if
and only if $(\mathrm{Gal}(L/K))^0$ is solvable.
\end{theorem}

Consider the differential equation
\begin{equation}\label{LDE}
\mathcal L:=\zeta''=r\zeta,\quad r\in K.
\end{equation}

We recall that equation \eqref{LDE} can be obtained from equation
\eqref{soldeq} through the change of variable
\begin{equation}\label{redsec}
y=e^{-{1\over 2}\int_{}^{}a}\zeta,\quad r={a^2\over
4}+{a'\over 2}-b
\end{equation} and equation \eqref{LDE} is called the \textit{reduced
form} (also known as \textit{invariant normal form}) of equation~\eqref{soldeq}.

On the other hand, introducing the change of variable
$v=\partial_x\zeta/\zeta$ we get the associated Riccati equation
to equation (\ref{LDE}),
\begin{equation}\label{Riccatti}
\partial_xv=r-v^2,\quad v={\zeta'\over \zeta},
\end{equation}
where $r$ is given by equation (\ref{redsec}). Moreover, the Riccatti equation \eqref{Riccatti} has one algebraic solution
over the differential field $K$ if and only if the differential
equation \eqref{LDE} is integrable.

For $\mathcal L$ given by equation
 (\ref{LDE}), it is very well
known that ${\rm Gal}_K(\mathcal L)$ is
an algebraic subgroup of ${\rm SL}(2,\mathbb{C})$. The well known
classification of subgroups of $\mathrm{SL}(2,\mathbb{C})$  is the following.

\begin{theorem}\label{subgroups} Let $G$ be an algebraic subgroup of ${\rm SL}(2,\mathbb{C})$.
Then, up to conjugation, one of the following cases occurs.
\begin{enumerate}
\item $G\subseteq \mathbb{B}$ and then $G$ is reducible and
triangularizable.
\item $G\nsubseteq\mathbb{B}$, $G\subseteq \mathbb{D}_\infty$ and then $G$ is imprimitive.
\item $G\in\{A_4^{\mathrm{SL}_2},S_4^{\mathrm{SL}_2},A_5^{\mathrm{SL}_2}\}$ and then $G$ is primitive (finite)
\item $G = {\rm SL}(2,\mathbb{C})$ and then $G$ is primitive (infinite).
\end{enumerate}
\end{theorem}

\subsection{Kovacic's Algorithm}
In 1986, Kovacic (\cite{kovacic:1986}) introduced an
algorithm to solve the differential equation (\ref{LDE}), where $K=\mathbb{C}(x)$, showing
that (\ref{LDE}) is integrable if and only if the solution of the
Riccati equation (\ref{Riccatti}) is a rational function (case 1),
is a root of polynomial of degree two (case 2) or is a root of
polynomial of degree 4, 6, or 12 (case 3). We leave the details of the algorithm to
Appendix A. We summarize here the main result by Kovacic as the following theorem.
\begin{theorem}[Kovacic]\label{Kov}
There are precisely four cases that can occur for equation~\eqref{LDE}:
\begin{description}
\item[Case 1] It has a solution of the form $e^{\int\omega}$ where $\omega\in\mathbb{C}(x)$.
\item[Case 2] It has a solution of the form $e^{\int\omega}$ where $\omega$ is algebraic over $\mathbb{C}(x)$ of degree 2, and case 1 does not hold.
\item[Case 3] All solutions of~\eqref{LDE} are algebraic over $\mathbb{C}(x)$ of degree 4, 6 or 12 and cases 1 and 2 do not hold.
\item[Case 4] The differential equation~\eqref{LDE} has no Liouvillian solution.
\end{description}
\end{theorem}
In the following sections we will apply the algorithm to equations~\eqref{eq:PDE11} and~\eqref{eq:PDE12} using a Laplace transform acting over the time variable.

\section{Linear rates: $(b,d)=(1,1)$}
\label{sec:11}

As mentioned in Section~\ref{sec:master}, we introduce the generating function $g(z,t)=\sum_{N=0}^{\infty}P_N(t)z^N$
to transform the discrete variable $N$ into the continuous variable $z$, $0\le z \le 1$. This converts the master equation
into a PDE: if we multiply both sides of the master equation~\eqref{eq:master} by $z^N$ and sum over $N$, we get
\begin{equation}\label{eq:BD}
\frac{\partial g(z,t)}{\partial t}=
\sum_{N=0}^{\infty}\left\{\beta(N-1)P_{N-1}(t)z^N+\delta(N+1)P_{N+1}(t)z^N-(\beta+\delta)NP_N(t)z^N\right\}.
\end{equation}
Recall that $P_N(t):=0$ for $N<0$. We now use the following straightforward identities:
\begin{itemize}
\item[(i)] $\frac{\partial g}{\partial z}=\sum_{N=0}^{\infty} (N+1)P_{N+1}(t)z^{N}=\sum_{N=1}^{\infty} NP_N(t)z^{N-1}$,
\item[(ii)] $z\frac{\partial g}{\partial z}=\sum_{N=1}^{\infty} NP_N(t)z^N$,
\item[(iii)] $z^2\frac{\partial g}{\partial z}=\sum_{N=1}^{\infty} (N-1)P_{N-1}(t)z^N$,
\item[(iv)] $\sum_{N=0}^{\infty} P_{N-1}(t)z^N=\sum_{N=1}^{\infty} P_{N-1}(t)z^N=\sum_{N=0}^{\infty} P_N(t)z^{N+1}=zg(z,t)$,
\end{itemize}
to get the first-order PDE~\eqref{eq:PDE11},
\begin{equation}
\frac{\partial g(z,t)}{\partial t}=\left\{\beta z^2-(\beta+\delta)z+\delta\right\}\frac{\partial g(z,t)}{\partial z}. \label{eq:PDE11b}
\end{equation}
The initial condition $P_N(0)=\delta_{N_0N}$ reduces to $g(z,0)=z^{N_0}$. Normalization of the probability distribution
at any time, $\sum_{N=0}^{\infty}P_N(t)=1$, implies that $g(1,t)=1$. Then we have to solve~\eqref{eq:PDE11b}
with the boundary conditions $g(z,0)=z^{N_0}$ and $g(1,t)=1$.

\subsection{Solution via a Laplace transform}
\label{sec:Lap11}

Although this case leads to a first-order PDE, which can be solved via the method of characteristics (see reference~\cite{kendall:1948} for the application of this method to 
the $(1,1)$ case in a more general setting in which the coefficients $\beta$ and $\delta$ are functions of time),
we calculate here the solution explicitly via the Laplace transform method to illustrate our methodology. Previously, however, it is convenient to clarify what kind of integrability we are considering in this work.

Let
\begin{equation}\frac {\partial g }{\partial t}=Mg \label{eq:PDE gen}
\end{equation}
be a partial differential equation, $M$ being a linear differential operator in the one-dimensional spatial variable $z$.  Then we can state the following
\medskip

\noindent{\bf Problem}. {\it Solve the PDE \eqref{eq:PDE gen} subject to suitable boundary conditions, including the initial Cauchy problem $g(z,0)=g_0(z)$. }

\medskip

Applying the Laplace transform {\it with respect to time} to \eqref{eq:PDE gen}, we obtain a family of linear ODE equations,
\begin{equation}
MG=sG+g_0(z),\label{eq:ODE gen}
\end{equation}
parameterized by the complex parameter $s$. We will say that equation \eqref{eq:ODE gen} is integrable if the homogeneous equation
$$MG=sG,$$
is integrable in the sense of Picard-Vessiot theory. This is natural, because from the general solution of the homogeneous equation we obtain the general solution of  \eqref{eq:ODE gen} by quadratures. Another approach to the integrability of \eqref{eq:ODE gen} is by transforming it to an homogeneous equation: later we will point out  an explicit example of this point of view. Of course, we are assuming here that the coefficients of $M$, and the function $g_0(z)$ belong to a suitable differential field $K$, for instance, the set of complex rational functions. Then

\begin{definition}\label{def:integ} We say that the equation \eqref{eq:PDE gen} is {\it integrable} if the family of linear   ODE equations \eqref{eq:ODE gen} is integrable in the sense of the Picard-Vessiot theory for almost any complex $s$.\end{definition}

We remark that despite it is usually assumed that the Laplace transformed function of the variable $s$ is defined in some half plane of the complex variable $s$, we are assuming here that this function can be prolongated analytically to other values of $s$.

Now focus on the PDE~\eqref{eq:PDE11}. It is clearly integrable, according to definition~\ref{def:integ}, because the associated linear ODE~\eqref{eq:ODE gen} is a first order ODE, being the coefficient field $K=\mathbb{C}(x)$ ---indeed, along the rest of the paper we will assume that the coefficient field is the set of rational functions $\mathbb{C}(x)$. Hence, we introduce the Laplace transform acting over the time dependence of the generating function as
\begin{equation}
G(z,s)=\int_0^{\infty} g(z,t)e^{-st}dt,
\end{equation}
which transforms~\eqref{eq:PDE11} into the following first-order ODE,
\begin{equation}\label{eq:ODE11}
(1-z)(\delta-\beta z)G'(z,s)=sG(z,s)-z^{N_0},
\end{equation}
where we regard $s$ as a parameter and primes denote derivatives with respect to $z$.  We now focus on
solving equation~\eqref{eq:ODE11} for arbitrary values of $s$. Note also that equation~\eqref{eq:ODE11} can be expressed as
\begin{equation}\label{eq:ODE11a}
G'(z,s)=f(z,s)G(z,s)+h(z)
\end{equation}
with 
\begin{equation}\label{eq:fh}
\begin{aligned}
&f(z,s)=\frac{s}{(1-z)(\delta-\beta z)},\\
&h(z)=-\frac{z^{N_0}}{(1-z)(\delta-\beta z)}.
\end{aligned}
\end{equation}
This form of the ODE will be convenient later in our computations. We observe that, in the homogeneous part of equation~\eqref{eq:ODE11a}, the point $z=\infty$ is an ordinary point. Moreover, when $\beta\neq \delta$ the points $z=1$ and $z=\frac{\delta}{\beta}$ are regular singular points, while when $\beta=\delta$ the point $z=1$ is a singularity of irregular type, see~\cite{almp2} for a detailed explanation about differential Galois theory of non-homogeneous equations.

From now on we shall consider these two cases ($\beta\ne \delta$ and $\beta=\delta$) separately. For $\beta\neq \delta$, the homogeneous equation can be solved
immediately,
\begin{equation}\label{eq:integ11}
\frac{G'}{G}=\frac{s}{(1-z)(\delta-\beta z)},\quad \ln G=\ln C+s \int \frac{dz}{(1-z)(\delta-\beta z)}.
\end{equation}
Assume that $\delta > \beta$ (the calculations for the $\delta<\beta$ case are simple extensions of those
provided here and are therefore left to Appendix B). Then equation~\eqref{eq:integ11} yields
\begin{equation}
\ln G=\ln C+\frac{s}{\delta-\beta}\int \left(\frac{1}{1-z}-\frac{\beta}{\delta-\beta z}\right)dz,
\end{equation}
i.e.,
\begin{equation}\label{eq:sol11}
G(z,s)=C\left(\frac{\delta-\beta z}{1-z}\right)^{\frac{s}{\delta-\beta}},
\end{equation}
where $C$ is an integration constant. Variation of the constant in equation~\eqref{eq:ODE11} yields a first-order ODE
for the unknown function $C(z)$,
\begin{equation}
(1-z)(\delta-\beta z)C'(z)\left(\frac{\delta-\beta z}{1-z}\right)^{\frac{s}{\delta-\beta}}=-z^{N_0},
\end{equation}
for which we impose $C(1)=0$ to avoid possible divergences in the generating function $g(z,t)$ at $z=1$ ---recall that $g(z,t)$ has to be an analytic function of $z$ because the probability distribution $P_N(t)$ is to be determined through a series expansion of $g(z,t)$ about $z=0$, see equation~\eqref{eq:genfun}. Therefore,
\begin{equation}\label{eq:C}
C(z)=\int_z^1 \frac{(1-u)^{\frac{s}{\delta-\beta}-1}}{(\delta-\beta u)^{\frac{s}{\delta-\beta}+1}}\,u^{N_0}du.
\end{equation}
Using~\eqref{eq:sol11} and~\eqref{eq:C} together, the Laplace transform of the generating function is expressed as
\begin{equation}
G(z,s)=\int_z^1 \left(\frac{\delta-\beta z}{1-z}\right)\left[\left(\frac{\delta-\beta z}{1-z}\right)\left(\frac{1-u}{\delta-\beta u}\right)\right]^{\frac{s}{\delta-\beta}-1}\frac{u^{N_0}}{(\delta-\beta u)^2}\,du. \label{eq:lapl11}
\end{equation}
In terms of the new variable $w(u):=\alpha\left(\frac{1-u}{\delta-\beta u}\right)$ with $\alpha:=\frac{\delta-\beta z}{1-z}$,
the integral above can be written as
\begin{equation}
G(z,s)=\frac{1}{\delta-\beta}\int_0^1 w^{\frac{s}{\delta-\beta}-1}\left(\frac{\alpha-w\delta}{\alpha-w\beta}\right)^{N_0}dw.
\end{equation}
After a second change of variable, $w(t):=e^{(\beta-\delta)t}$, we finally get
\begin{equation}
G(z,s)=\int_0^{\infty} \left(\frac{\alpha-w(t)\delta}{\alpha-w(t)\beta}\right)^{N_0}e^{-st}dt,
\end{equation}
which allows us to identify the generating function
\begin{equation}\label{eq:genfun1}
g(z,t)=\left(\frac{\alpha-w(t)\delta}{\alpha-w(t)\beta}\right)^{N_0}=\left[\frac{\delta-\beta z-(1-z)\delta e^{(\beta-\delta)t}}{\delta-\beta z-(1-z)\beta e^{(\beta-\delta)t}}\right]^{N_0}.
\end{equation}
Integration of~\eqref{eq:PDE11} for $\beta<\delta$ yields exactly the same expression (see Appendix B).

Now, considering $\beta=\delta$, equation~\eqref{eq:integ11} yields
\begin{equation}
\ln G=\ln C+\frac{s}{\delta}\int dz\,\frac{1}{(1-z)^2},
\end{equation}
i.e.,
\begin{equation}\label{eq:sol11bd}
G(z,s)=Ce^{\frac{s}{\delta}\frac{1}{1-z}},
\end{equation}
where $C$ is again an integration constant. Variation of the constant gives again a first-order ODE
for $C(z)$,
\begin{equation}
\delta(1-z)^2C'(z)e^{\frac{s}{\delta}\frac{1}{1-z}}=-z^{N_0},
\end{equation}
for which we impose $C(1)=0$ to avoid divergences, as above. Therefore, the general solution of equation \eqref{eq:ODE11} is
\begin{equation}\label{soleq:ODE11}
 G(z,s)=\int_z^1 \frac {u^{N_0}}{\delta(1-u)^2} e^{-\frac s{\delta}\left(\frac 1{1-u}-\frac 1{1-z}\right)}du.
\end{equation}

The previous function can be obtained through iterated partial integration and, for $N_0\in\mathbb{Z^+}$, the result belongs to the family of \textit{exponential integrals}, denoted by Ei, which is valid for $\mathfrak{R}(z)>0$ ---as in our case because $0\leq z\leq 1$. Ei functions are not elementary functions, see \cite{abst} for further details. But in fact, we are interested here in the inverse-Laplace transformed function, $g(z,t)$, that becomes an elementary function. So, by means of the change $t(u)=\frac 1{\delta}\big(\frac 1{1-u}-\frac 1{1-z}\big)$, we obtain
\begin{equation}\label{lapsolbeqd}
 G(z,s)=\int_0^\infty \left[\frac {\delta(1-z)t+z}{\delta(1-z)t+1}\right]^{N_0} e^{- st}dt.
\end{equation}
Then
\begin{equation}\label{genebeq}
 g(z,t)=\left[\frac {\delta(1-z)t+z}{\delta(1-z)t+1}\right]^{N_0},
\end{equation}
is the sought solution of~\eqref{eq:PDE11} for $\beta=\delta$, satisfying  the boundary conditions $g(z,0)=z^{N_0}$ and $g(1,t)=1$.

In summary, proposition \ref{prop 11} has been proved. We consider the $(b,d)=(1,1)$ case as completely solved since the probability distribution $P_N(t)$ could eventually be obtained through a series expansion of the generating function. In particular, useful expressions for the mean and the variance of the distribution (or even any moment) can be computed for arbitrary values of $N$ and $t$. In addition, the probability of extinction at time $t$ is given by
\begin{equation}\label{eq:P0}
g(0,t)=P(0,t)=\left[\frac{\delta\left(e^{(\beta-\delta)t}-1\right)}{\beta e^{(\beta-\delta)t}-\delta}\right]^{N_0}
\end{equation}
for $\beta\neq \delta$, and
\begin{equation}\label{eq:P0}
g(0,t)=P(0,t)=\left(\frac{\delta t}{\delta t+1}\right)^{N_0}
\end{equation}
for $\beta=\delta$.

\subsection{Solution via Kovacic's algorithm}
\label{sec:Kov11}

For quadratic death rates we find a second-order PDE for the generating function, see equation~\eqref{eq:PDE12} and Section~\ref{sec:12}. The integrability
of this case can be analyzed using Kovacic's algorithm~\cite{kovacic:1986} since the Laplace transform yields a second-order, linear ODE
whose coefficients are rational functions. As we anticipated in Proposition~\ref{prop 12}, the $(b,d)=(1,2)$ PDE is not integrable.
However, we have just shown that, for the linear-rate case $(b,d)=(1,1)$, the problem is integrable. Kovacic's
algorithm usually restricts the values of the parameters in the differential equation in order to yield integrability. In both cases, the
Laplace transform method introduces a new parameter in the equations ---the parameter $s$ associated to the time dependence.

In this section we apply the algorithm by Kovacic to the $(b,d)=(1,1)$ case in order to gain some insight about integrability
of the PDE via the Laplace transform: obviously,  we have to recover the solution~\eqref{eq:sol11} with no restrictions imposed by the
algorithm on the Laplace transform parameter $s$, in agreement with our definition of integrability.

We can apply Kovacic's algorithm to a second-order, linear ODE whose coefficients are rational functions. In order to apply Kovacic's algorithm to the
inhomogeneous, first-order ODE~\eqref{eq:ODE11}, we transform the equation as follows: first eliminate the first derivative,
\begin{equation}
G'(z,s)=\frac{s}{(1-z)(\delta-\beta z)}G(z,s)-\frac{z^{N_0}}{(1-z)(\delta-\beta z)}, \label{eq:ODE11b}
\end{equation}
and then divide the equation by the term $\frac{z^{N_0}}{(1-z)(\delta-\beta z)}$ to get
\begin{equation}
\frac{(1-z)(\delta-\beta z)}{z^{N_0}}G'(z,s)=\frac{s}{z^{N_0}}G(z,s)-1.
\end{equation}
Differentiating both sides of the equation above yields a second-order, linear, homogeneous equation whose coefficients are rational
functions of $z$:
\begin{equation}\label{eq:ODE11K}
G''(z,s)-\frac{(N_0-2)\beta z^2+[s-(N_0-1)(\delta+\beta)]z+\delta N_0}{z(1-z)(\delta-\beta z)}G'(z,s)+\frac{s N_0}{z(1-z)(\delta-\beta z)}G(z,s)=0.
\end{equation}

Now it is convenient to clarify the relation between the solutions of the linear equation \eqref{eq:ODE11b} and of the second order \eqref{eq:ODE11K}, that we write as a lemma for future reference.
\begin{lemma}\label{lema11}
Consider a first order linear ODE,
\begin{equation}
G'=fG+h, \label{eq:lorder1}
\end{equation}
with general solution
\begin{equation}
G_1=C_1e^{\int f dz}+e^{\int f dz}\int e^{-\int f}h dz. \label{eq:sollorder1}
\end{equation}
Then the general solution of the associated second order, linear ODE obtained by derivation over equation~\eqref{eq:lorder1} divided by $h$,
\begin{equation}
G''-\left(f+\frac{h'}h\right)G'+\left(f\frac{h'}h-f'\right)G=0, \label{eq:lorder2}
\end{equation}
is given by
\begin{equation}
G_2=C_1e^{\int f dz}+C_2e^{\int f dz}\int e^{-\int f}h dz, \label{eq:sollorder2}
\end{equation}
\end{lemma}
\begin{proof}
A first integral of equation \eqref{eq:lorder2} is given by the linear first order equation 
\begin{equation}
\frac{G'-f G}{h}=:C_2\Leftrightarrow  G'=fG+C_2h, \label{eq:lorder1g}
\end{equation}
which coincides with \eqref{eq:lorder1} for $C_2=1$. Then solving equation \eqref{eq:lorder1g}, we obtain \eqref{eq:sollorder2}.
 \end{proof}
In other words, a fundamental system of solutions of \eqref{eq:lorder2} is given by a non-trivial solution of the homogeneous part of \eqref{eq:lorder1} and by any of the  particular solutions of \eqref{eq:lorder1} (like the one obtained by variation of constants). In particular, \eqref{eq:lorder2} has always a solution  given by the exponential of an integral: $e^{\int f dz}$.
\medskip

Now we focus on the solutions of the second-order ODE~\eqref{eq:ODE11K} yielded by Kovacic's algorithm. For that purpose we normalize~\eqref{eq:ODE11K} to write it in the form $H''-r(z,s)H=0$
for a new function $H(z,s)$. If we define
\begin{equation}\label{eq:ab}
\begin{aligned}
&a(z,s):=-\frac{(N_0-2)\beta z^2+[s-(N_0-1)(\delta+\beta)]z+\delta N_0}{z(1-z)(\delta-\beta z)},\\
&b(z,s):=\frac{s N_0}{z(1-z)(\delta-\beta z)},
\end{aligned}
\end{equation}
then the invariant normal form of~\eqref{eq:ODE11K} is obtained using equation~\eqref{redsec}:
\begin{equation}\label{eq:ODE11K2}
H''(z,s)-\left(\frac{1}{2}a'(z,s)+\frac{1}{4}a^2(z,s)-b(z,s)\right)H(z,s)=0,
\end{equation}
where $G(z,s)=H(z,s)\psi(z,s)$ and $\psi(z,s)$ satisfies the first-order ODE
\begin{equation}\label{eq:psieq}
2\psi'(z,s)+a(z,s)\psi(z,s)=0.
\end{equation}
Note also that, for $\beta\neq\delta$,
\begin{equation}\label{eq:a}
a(z,s)=-\frac{N_0}{z}-\left(1+\frac{s}{\delta-\beta}\right)\frac{1}{1-z}-\left(1-\frac{s}{\delta-\beta}\right)\frac{\beta}{\delta-\beta z},
\end{equation} 
while for $\beta=\delta$,
 \begin{equation}\label{eq:abd}
 a(z,s)=-\frac{N_0}{z}-\frac{2}{1-z}+\frac{s}{\delta}\frac{1}{(1-z)^2}.
\end{equation}
Integration of~\eqref{eq:psieq} yields
\begin{equation}\label{eq:psi}
\psi(z,s)=z^{N_0/2}(1-z)^{-\frac{1}{2}\left(1+\frac{s}{\delta-\beta}\right)}(\delta-\beta z)^{-\frac{1}{2}\left(1-\frac{s}{\delta-\beta}\right)}, \,\, \beta\ne\delta,
\end{equation}
and \begin{equation}\label{eq:psibd}
\psi(z,s)=z^{N_0/2}(1-z)^{-1}e^{-\frac{s}{2\delta}\frac{1}{(1-z)}}, \,\,\beta=\delta.
\end{equation}
Now we apply Kovacic's algorithm to~\eqref{eq:ODE11K} to check the integrability of  this equation. Together
with~\eqref{eq:psi}, we will construct solutions for the Laplace transform of the generating function as $G(z,s)=H(z,s)\psi(z,s)$. We recall that, by Lemma \ref{lema11}, the equation~\eqref{eq:ODE11K} has always a solution given by the exponential of an integral of a rational function,
$$G=e^{\int f dz}.$$ 
Then equation~\eqref{eq:ODE11K} has a solution given by the exponential of an integral in $K$,
$$H=G\psi^{-1}=e^{\int \left(f+\frac a2\right)dz}.$$
This implies that~\emph{case 1 of Kovacic's algorithm always holds for equation~\eqref{eq:ODE11K2}.}

The computations that lead to the closed-form solution of equation~\eqref{eq:ODE11K2} go as follows (see Appendix A for details on how the algorithm proceeds in a general setup). As can be easily checked, the rational function
\begin{equation}\label{eq:r11}
r(z,s)=\frac{1}{2}a'(z,s)+\frac{1}{4}a^2(z,s)-b(z,s)
\end{equation}
has three finite singularities at $z=0$, $z=1$ and $z=\delta/\beta$ if $\beta\ne\delta$. The algorithm is based on the
orders of the poles of $r(z,s)$ in the complex plane, considering the singularity $z=\infty$ as well. Let $\Gamma'$ be the set of
finite poles of $r(z,s)$ in the complex plane. Let $\Gamma=\Gamma'\cup\{\infty\}$. The method is based on the Laurent series
expansions of $r(z,s)$ about the singularities in $\Gamma$ (see Appendix A). Let $\circ(c)$ denote the order of the pole $c$ in the Laurent
series expansion. In this example, the following series expansions hold:
\begin{itemize}
\item[(i)] $r(z,s)=\frac{N_0(N_0+2)}{4z^2}+\dots$ about $z=0$.
\item[(ii)] $r(z,s)=\frac{1}{4}\left(-1+\frac{s^2}{(\delta-\beta)^2}\right)\frac{1}{(z-1)^2}+\dots$ about $z=1$.
\item[(iii)] $r(z,s)=\frac{1}{4}\left(-1+\frac{s^2}{(\delta-\beta)^2}\right)\frac{1}{(z-\delta/\beta)^2}+\dots$ about $z=\frac{\delta}{\beta}$.
\item[(iv)] $r(z,s)=\frac{N_0(N_0-2)}{4z^2}+\dots$ about $z=\infty$.
\end{itemize}
We study the existence of case 1 solutions in Kovacic's algorithm: all the poles have order 2, hence $\circ(c)=2$ for all $c\in\Gamma$. Therefore we write the 
Laurent series expansion of $\sqrt{r}$ about $c$, $[\sqrt{r}]_c$, as $[\sqrt{r}]_c=0$ for all $c\in\Gamma$ (see details about the general notation used in case 1 in 
Appendix A). Then we compute
\begin{equation}
\alpha_c^{\pm}=\frac{1}{2}\pm\frac{1}{2}\sqrt{1+4b},
\end{equation}
where $b$ is the residue of $r$ at the singularity $c$ (Appendix A). For $z=0$ we obtain $\alpha_0^{+}=1+\frac{N_0}{2}$ and
$\alpha_0^{-}=-\frac{N_0}{2}$. For $z=1$, we get $\alpha_1^{\pm}=\frac{1}{2}\left(1\pm \frac{s}{\delta-\beta}\right)$.
For $z=\delta/\beta$ we obtain $\alpha_{\delta/\beta}^{\pm}=\alpha_1^{\pm}$ because the residues associated to both
singularities coincide. Finally, for $z=\infty$ we obtain $\alpha_{\infty}^{+}=\frac{N_0}{2}$ and $\alpha_0^{-}=1-\frac{N_0}{2}$.

We now form the $2^4$ possible permutations of signs for the four singularities and compute the quantity 
$m=\alpha_{\infty}^{\varepsilon(\infty)}-\sum_{c\in\Gamma'}\alpha_c^{\varepsilon(c)}$. Let $\hat{s}:=\frac{s}{\delta-\beta}$. Then
\begin{table*}[h!]
\begin{center}
\begin{tabular}{ccccc}
\hline
$\varepsilon(\infty)$ & $\varepsilon(\delta/\beta)$ & $\varepsilon(1)$ & $\varepsilon(0)$ & $m=\alpha_{\infty}^{\varepsilon(\infty)}-\sum_{c\in\Gamma'}\alpha_c^{\varepsilon(c)}$\\
\hline
$+$ & $\pm$ & $\pm$ & $+$ & $-2\mp\hat{s}$\\
$+$ & $\pm$ & $\pm$ & $-$ & $N_0-1\mp\hat{s}$\\
$-$ & $\pm$ & $\pm$ & $+$ & $-N_0-1\mp\hat{s}$\\
$-$ & $\pm$ & $\pm$ & $-$ & $\mp\hat{s}$\\
$+$ & $\pm$ & $\mp$ & $+$ & $-2$\\
$+$ & $\pm$ & $\mp$ & $-$ & $N_0-1$\\
$-$ & $\pm$ & $\mp$ & $+$ & $-N_0-1$\\
$-$ & $\pm$ & $\mp$ & $-$ & $0$\\
\hline
\end{tabular}
\end{center}
\end{table*}

According to the algorithm, we have to consider only those permutations that yield a non-negative integer $m$ (Appendix A). This discards, for example, the cases
$(\varepsilon(\infty),\varepsilon(\delta/\beta),\varepsilon(1),\varepsilon(0))=(+,\pm,\mp,+)$ and $(-,\pm,\mp,+)$ ---recall that $N_0$ is a non-negative integer number. 
It is possible to find integrability
for the $(+,\pm,\mp,-)$ and $(-,\pm,\mp,-)$ cases. The remaining $8$ cases depend explicitly on $s$, hence
imposing that $m$ is a non-negative integer would restrict the possible values of $s$ yielding closed-form solutions. In principle
we are interested in finding a solution valid for \emph{any} value of $s$, so we do not enter in the discussion of these cases in
the main text. We leave the analysis of some of them for Appendix C, including computations for the sign combinations $(+,\pm,\mp,-)$.

Here we will discuss only one of the potential $4$ cases that can yield a solution: consider the permutation $(-,+,-,-)$, which corresponds to $m=0$. Then the
algorithm proceeds by considering the rational function
\begin{equation}
\omega(z,s)=\sum_{c\in\Gamma'}\left(\varepsilon(c)[\sqrt{r}]_c+\frac{\alpha_c^{\varepsilon(c)}}{z-c}\right)+s(\infty)[\sqrt{z}]_{\infty}.
\end{equation}
In our example this function reduces to
\begin{equation}
\omega(z,s)=\frac{\alpha_0^-}{z}+\frac{\alpha_1^-}{z-1}+\frac{\alpha_{\delta/\beta}^+}{z-\delta/\beta}.
\end{equation}
Given that $\alpha_0^-=-\frac{N_0}{2}$, $\alpha_1^-=\frac{1}{2}\left(1-\frac{s}{\delta-\beta}\right)$ and
$\alpha_{\delta/\beta}^+=\frac{1}{2}\left(1+\frac{s}{\delta-\beta}\right)$ we get, according to equation~\eqref{eq:a},
\begin{equation}\label{eq:omega}
\omega(z,s)=-\frac{N_0}{2z}-\frac{1}{2}\left(1-\frac{s}{\delta-\beta}\right)\frac{1}{1-z}-
\frac{1}{2}\left(1+\frac{s}{\delta-\beta}\right)\frac{\beta}{\delta-\beta z}.
\end{equation}
The algorithm now searches for a monic polynomial $P_m(z)$ of degree $m$ that satisfies the differential equation
\begin{equation}\label{eq:Pm}
P_m''+2\omega P_m'+(\omega'+\omega^2-r)P_m=0.
\end{equation}
If such polynomial exists, then a solution of the form $P_me^{\int\omega}$ exists. In our case
$m=0$ and, as it can be easily checked using~\eqref{eq:r11}, the function $\omega(z,s)$ defined in equation~\eqref{eq:omega}
satisfies identically the condition $\omega'+\omega^2-r=0$. Therefore equation~\eqref{eq:Pm} is satisfied by the constant
monic polynomial $P_0=1$ and we find the following closed-form solution for~\eqref{eq:ODE11K2}:
\begin{equation}
H(z,s)=\exp\left\{\int^z \omega(u,s)\,du\right\}=
z^{-N_0/2}(1-z)^{\frac{1}{2}\left(1-\frac{s}{\delta-\beta}\right)}(\delta-\beta z)^{\frac{1}{2}\left(1+\frac{s}{\delta-\beta}\right)}.
\end{equation}
Therefore, using~\eqref{eq:psi}, we get
\begin{equation}\label{eq:G1}
G(z,s)=H(z,s)\psi(z,s)=\left(\frac{\delta-\beta z}{1-z}\right)^{\frac{s}{\delta-\beta}}
\end{equation}
and we recover the solution~\eqref{eq:sol11} obtained from the first-order homogeneous ODE. Note that 
$G(z,s)=\left(\frac{\delta-\beta z}{1-z}\right)^{\frac{s}{\delta-\beta}}=e^{\int f(z,s)dz}$, with $f(z,s)$ given by~\eqref{eq:fh}.
But this solution is not the Laplace transform $G(z,s)$ ---equation~\eqref{eq:lapl11}--- of the generating function $g(z,t)$ 
we are looking for. However, we can use~\eqref{eq:G1} and Lemma~\ref{lema11} with $f$ and $h$ given by~\eqref{eq:fh} 
to construct the relevant solution as
\begin{equation}
G(z,s)=e^{\int f dz}\int e^{-\int f}h dz=-\left(\frac{\delta-\beta z}{1-z}\right)^{\frac{s}{\delta-\beta}}\int^z
\left(\frac{1-u}{\delta-\beta u}\right)^{\frac{s}{\delta-\beta}}\frac{u^{N_0}}{(1-u)(\delta-\beta u)}\,du,
\end{equation}
which is the exact same solution obtained in~\eqref{eq:lapl11}. Alternatively, it would be also possible to obtain this solution by applying to 
equation~\eqref{eq:ODE11K} the D'Alambert order reduction of a linear equation when a particular solution is known ---we, however, skip the details here.

In a similar way we apply Kovacic's algorithm for $\beta=\delta$. Now, equation \eqref{eq:r11} becomes
$$r(z,s)=\frac{N_0(N_0+2)}{4z^2}+\frac{N_0(2\delta+s)}{2\delta(1-z)}+\frac{N_0s}{2\delta(1-z)^2}+\frac{N_0(2\delta+s)}{2\delta z}+\frac{s^2}{4\delta^2(1-z)^4}.$$
Applying the case 1 of Kovacic's algorithm we obtain that the solution of $H''=r(z,s)H$ is
$$H(z,s)=z^{-\frac{N_0}{2}}(1-z)e^{-\frac{s}{2\delta(1-z)}}.$$
Now, using equation \eqref{eq:psibd}, we conclude that
$$G(z,s)=H(z,s)\psi(z,s)=e^{\frac{s}{\delta(1-z)}},$$
as in~\eqref{eq:sol11bd}. We can recover the sought Laplace transform~\eqref{soleq:ODE11} using Lemma~\ref{lema11} as presented above.

An important insight that we infer thanks to the analysis of the first-order equation via Kovacic's algorithm is the
following conjecture: if we were to obtain integrability of the corresponding PDE  via a Laplace transform, we conjecture that a
necessary condition to obtain solutions of the form of Kovacic's first case is that the combination
\begin{equation}
m=\alpha_{\infty}^{\varepsilon(\infty)}-\sum_{c\in\Gamma'}\alpha_c^{\varepsilon(c)}
\end{equation}
remains \emph{independent} of $s$, as our definition of integrability for equation~\eqref{eq:PDE gen} requires integrability of the linear ODE~\eqref{eq:ODE gen} for any value of the parameter $s$.

\noindent{\bf Remark}: we observe that for $\beta\neq \delta$ equation~\eqref{eq:ODE11K} has 4 singular regular points at $z=0$, $z=1$, $z=\beta/\delta$ and $z=\infty$.
Therefore, it corresponds exactly to the general Heun's differential equation in the independent variable $z$ with parameters $\delta/\beta$, $sN_0/\beta$, $0$, $1-N_0$, $-N_0$, and  $(\beta-\delta+s)/(\beta-\delta)$. On the other hand, when $\beta=\delta$, we can observe that this equation has two regular singularities at $z=0$ and $z=\infty$, while it has one irregular singularity at $z=1$. We conclude that, with the changes of variables $G\mapsto Gz^{-N_0-1}/(1-z)^2$ and $z\mapsto (z-1)/z$, the equation corresponds to the confluent Heun's differential equation with parameters $s/\delta$, $N_0-1$, $N_0+1$, $0$, $(N_0^2\delta-sN_0+\delta)/(2\delta)$. Moreover, we observe that in the non-homogeneous first order linear differential equation the points $z=\infty$ and $z=0$ are ordinary points, but with the procedure to transform it into an homogeneous second order linear differential equation the points $z=\infty$ and $z=0$ are regular singular points. The type of singularity of $z=1$ and $z=\beta/\delta$ is preserved under such procedure for the cases $\beta=\delta$ and $\beta\neq \delta$, though. For further details about Heun's differential equations, we refer the reader to reference~\cite{ronveaux:1995}. We remark that a complete characterization of the integrability of Heun's equations is today an open problem. Here it was possible to solve the integrability problem because the equations correspond to very special subfamilies of Heun's general families.

\section{Mixed rates: $(b,d)=(1,2)$}
\label{sec:12}

As we have shown in the previous section, Kovacic's algorithm turns out to be a powerful tool to analyze the integrability of PDE associated to
birth-death processes via a Laplace transform. In biological terms, a relevant case arises when mortality processes involve pairs of individuals, i.e.,
when the death rate is a quadratic function of the number of individuals. In this section we apply the same technology to analyze the integrability
of the PDE associated to this situation.

As in the case of linear birth and death rates, we start by finding the PDE satisfied by the generating function when the
birth rate is linear, $B_N=\beta N$, and the mortality rate is a quadratic function of $N$, $D_N=\delta N^2$. The
generating function satisfies
\begin{equation}\label{eq:BD2}
\begin{aligned}
\frac{\partial g(z,t)}{\partial t}&=
\sum_{N=0}^{\infty}\left\{\beta(N-1)P_{N-1}(t)z^N+\delta(N+1)^2P_{N+1}(t)z^N-(\beta N+\delta N^2)P_N(t)z^N\right\}\\
&=\beta\sum_{N=1}^{\infty}\left[(N-1)P_{N-1}(t)z^N-NP_N(t)z^N\right]+\delta\sum_{N=0}^{\infty}\left[(N+1)^2P_{N+1}(t)z^N-N^2P_N(t)z^N\right],
\end{aligned}
\end{equation}
where we have used that $P_N(t):=0$ for $N<0$. The following identities hold:
\begin{itemize}
\item[(i)] $\frac{\partial g}{\partial z}+z\frac{\partial^2 g}{\partial z^2}=\sum_{N=1}^{\infty} N^2P_N(t)z^{N-1}=\sum_{N=0}^{\infty} (N+1)^2P_{N+1}(t)z^{N}$,
\item[(ii)] $z\frac{\partial g}{\partial z}=\sum_{N=1}^{\infty} NP_N(t)z^N$,
\item[(iii)] $z^2\frac{\partial g}{\partial z}=\sum_{N=1}^{\infty} (N-1)P_{N-1}(t)z^N$.
\end{itemize}
Therefore, we can express
\begin{equation}
\sum_{N=0}^{\infty}\left[(N+1)^2P_{N+1}(t)z^N-N^2P_N(t)z^N\right]=(1-z)\left(\frac{\partial g(z,t)}{\partial z}+z\frac{\partial^2 g(z,t)}{\partial z^2}\right)
\end{equation}
and
\begin{equation}
\sum_{N=1}^{\infty}\left[(N-1)P_{N-1}(t)z^N-NP_N(t)z^N\right]=-z(1-z)\frac{\partial g(z,t)}{\partial z}.
\end{equation}
As a consequence, we obtain a second-order PDE to be satisfied by the generating function, see equation~\eqref{eq:PDE12}.
Similarly, we impose here the initial condition $g(z,0)=z^{N_0}$ and the normalization condition $g(1,t)=1$. In order to find solutions of equation~\eqref{eq:PDE12}, we follow the same procedure as for the $(1,1)$ case: we introduce the Laplace transform $G(z,s)$ of the generating function and try to solve the parametric ODE satisfied by $G(z,s)$ for arbitrary values of $s$. In terms of $G(z,s)$, the ODE reads
\begin{equation}\label{eq:ODE12}
(1-z)\left[\delta z G''(z,s)+(\delta -\beta z)G'(z,s)\right]-sG(z,s)=-z^{N_0},
\end{equation}
where, again, primes denote derivatives with respect to $z$. Here we denote
\begin{equation}\label{eq:ab12}
\begin{aligned}
&a(z):=\frac{\delta-\beta z}{\delta z},\\
&b(z,s):=-\frac{s}{\delta z(1-z)},
\end{aligned}
\end{equation}
hence~\eqref{eq:ODE12} can be expressed as
\begin{equation}\label{eq:ODE12b}
G''(z,s)+a(z)G'(z,s)+b(z,s)G(z,s)=-\frac{z^{N_0}}{\delta z(1-z)}.
\end{equation}
In order to find the invariant normal form of~\eqref{eq:ODE12b}, we write $G(z,s)=H(z,s)\psi(z)$ and impose that $\psi(z)$ satisfies the first-oder ODE
\begin{equation}
2\psi'(z)+a(z)\psi(z)=0.
\end{equation}
Note that, in this case, $\psi$ is independent of $s$. Integration yields $\psi(z)=z^{-1/2}e^{\beta z/2\delta}$. Hence~\eqref{eq:ODE12b} reduces to
the following second-order, non-homogeneous ODE for function $H(z,s)$:
\begin{equation}
H''(z,s)\psi(z)-\left(\frac{1}{2}a'(z,s)+\frac{1}{4}a^2(z,s)-b(z,s)\right)H(z,s)\psi(z)=-\frac{z^{N_0}}{\delta z(1-z)}.
\end{equation}
Equivalently, $H(z,s)$ satisfies
\begin{equation}\label{eq:ODEH}
H''(z,s)-\left[\left(-\frac{1}{2z}+\frac{\beta}{2\delta}\right)^2-\frac{1}{2z^2}+\frac{s}{\delta z(1-z)}\right]H(z,s)=-\frac{z^{N_0-\frac{1}{2}}}{\delta(1-z)}e^{-\beta z/2\delta}.
\end{equation}
This is the second-order normal invariant form of the original ODE. We want to see whether we can find closed-form solutions for
 this ODE for \emph{any} value of the parameter $s$.

\noindent{\bf Remark}: this equation has 3 singular  points (as shown below) at $z=0$, $z=1$ (regular ones), and $z=\infty$ (irregular).
Therefore, it belongs to the family of Heun's confluent equations~\cite{ronveaux:1995}. Thus, the $(b,d)=(1,2)$ case also belongs to Heun's families.

\subsection{Solution via Kovacic's algorithm}

In line with our definition of integrability (definition~\ref{def:integ}), we look for closed-form solutions of the homogeneous part of 
equation~\eqref{eq:ODEH}. For that purpose we define
\begin{equation}\label{eq:r}
r(z,s)=\left(-\frac{1}{2z}+\frac{\beta}{2\delta}\right)^2-\frac{1}{2z^2}+\frac{s}{\delta z(1-z)}=\frac{\beta^2z^3 - \beta (\beta + 2 \delta)z^2 - \delta (4 s - 2 \beta + \delta)z + \delta^2 }{4\delta^2 z^2(z-1)}
\end{equation}
and apply Kovacic's algorithm to search for closed-form solutions. In our case, $\Gamma'=\{0,1\}$ with orders $\circ(0)=2$ and $\circ(1)=1$. The following series expansions for $r(z,s)$ about the three elements in $\Gamma$ hold:
\begin{itemize}
\item[(i)] $r(z,s)=-\frac{1}{4z^2}+\dots$ about $z=0$.
\item[(ii)] $r(z,s)=-\frac{s}{\delta(z-1)}+\dots$ about $z=1$.
\item[(iii)] $r(z,s)=\frac{\beta^2}{4\delta^2}-\frac{\beta}{2\delta z}+\dots$ about $z=\infty$.
\end{itemize}
We observe in equation~\eqref{eq:r} that the order of $r$ at $\infty$ is $\circ(\infty)=0$. We analyze the different cases in the algorithm by Kovacic 
(see further details in Appendix A):
\begin{description}
\item[Case 1] 
Since $\circ(0)=2$, we set $[\sqrt{r}]_0=0$ and obtain $\alpha_0^{\pm}=\frac{1}{2}\pm\frac{1}{2}\sqrt{1+4b}=\frac{1}{2}$ because the residue at $z=0$ is $b=-\frac{1}{4}$.

For $z=1$, since $\circ(1)=1$, we set $[\sqrt{r}]_1=0$ and $\alpha_1^{\pm}=1$.

For $z=\infty$, since $\circ(\infty)=0=-2\nu$ and $r(z,s)=q^2+b/z+\dots$, with $q=\frac{\beta}{2\delta}$ and $b=-\frac{\beta}{2\delta}$, we set $[\sqrt{r}]_{\infty}=q=\frac{\beta}{2\delta}$ and $\alpha_{\infty}^{\pm}=\left(\pm\frac{b}{q}-\nu\right)/2=\mp\frac{1}{2}$.

We now consider all the possible combinations of signs for the three points:

\begin{table*}[h]
\begin{center}
\begin{tabular}{cccc}
\hline
$\varepsilon(\infty)$ & $\varepsilon(0)$ & $\varepsilon(1)$ & $m=\alpha_{\infty}^{\varepsilon(\infty)}-\alpha_0^{\varepsilon(0)}-\alpha_1^{\varepsilon(1)}$\\
\hline
$+$ & $\pm$ & $\pm$ & $-\frac{1}{2}-\frac{1}{2}-1=-2$\\
$-$ & $\pm$ & $\pm$ & $\frac{1}{2}-\frac{1}{2}-1=-1$\\
$+$ & $\mp$ & $\mp$ & $-\frac{1}{2}-\frac{1}{2}-1=-2$\\
$-$ & $\mp$ & $\mp$ & $\frac{1}{2}-\frac{1}{2}-1=-1$\\
\hline
\end{tabular}
\end{center}
\end{table*}

Since all the values of $m$ are negative integers, Kovacic's algorithm does not find solutions of the form $P_m e^{\int \omega}$ with 
$P_m$ a polynomial.

\item[Case 2] Given the orders of the singularities of $r(z,s)$, we define the following subsets of $\mathbb{Z}$ (see a full description of how the algorithm 
proceeds in this case in Appendix A):

For $z=0$, since $\circ(0)=2$ and the residue at $z=0$ is $b=-\frac{1}{4}$, we have $E_0:=\{2+k\sqrt{1+4b}, k=0,\pm 2\}=\{2\}$.

For $z=1$, since $\circ(1)=1$, we define $E_1:=\{4\}$.

For $z=\infty$, since $\circ(\infty)=0<2$, then $E_{\infty}:=\{0\}$.

We now find the positive combinations of the sum $m=\frac{1}{2}\left(e_{\infty}-\sum_{c\in\Gamma'}e_c\right)$ for $e_p\in E_p$, $p\in\Gamma$. The only combination is $(e_0,e_1,e_{\infty})=(2,4,0)$, hence $m=\frac{1}{2}(0-2-4)<0$. Therefore the set of positive $m$ is empty and there are no solutions in this case.

\item[Case 3] A necessary condition for this case to work is that $\circ(\infty)\ge 2$, see~\cite{kovacic:1986}. There are no solutions of this type since 
$\circ(\infty)=0$.

\end{description}

Therefore, we conclude that  equation~\eqref{eq:ODEH} is non-integrable for any value of $s$. Hence,  the homogeneous part of  equation~\eqref{eq:ODE12} is also  non-integrable and, as a consequence, the PDE~\eqref{eq:PDE12} becomes non-integrable as well. This proves proposition~\ref{prop 12}.

\appendix

\section*{Appendix A}
\renewcommand\theequation{B.\arabic{equation}}

This Appendix describes Kovacic's algorithm in detail. In our presentation here, we follow the original
version given by Kovacic in reference \cite{kovacic:1986} with an adapted
version presented in \cite{acbl,amw}.

Each case in Kovacic's algorithm is related with each one of the
algebraic subgroups of ${\rm SL}(2,\mathbb{C})$ and the associated
Riccatti equation
$$v'=r-v^{2}=\left( \sqrt{r}-v\right)
\left(  \sqrt{r}+v\right),\quad v={\zeta'\over \zeta}.$$

According to Theorem \ref{subgroups}, there are four cases in
Kovacic's algorithm. Only for cases 1, 2 and 3 we can solve the
differential equation, but for the case 4 the differential
equation is not integrable. It is possible that Kovacic's
algorithm can provide us only one solution ($\zeta_1$), so that we
can obtain the second solution ($\zeta_2$) through
\begin{equation}\label{second}
\zeta_2=\zeta_1\int\frac{dx}{\zeta_1^2}.
\end{equation}

\noindent{\bf\large Notations.} For the differential equation given by
$$\partial_x^2\zeta=r\zeta,\qquad r={s\over t},\quad s,t\in \mathbb{C}[x],$$
we use the following notations.
\begin{enumerate}
\item Denote by $\Gamma'$ be
the
set of (finite) poles of $r$, $\Gamma^{\prime}=\left\{  c\in\mathbb{C}%
:t(c)=0\right\}$.

\item Denote by
$\Gamma=\Gamma^{\prime}\cup\{\infty\}$.
\item By the order of $r$ at
$c\in \Gamma'$, $\circ(r_c)$, we mean the multiplicity of $c$ as a
pole of $r$.

\item By the order of $r$ at $\infty$, $\circ\left(
r_{\infty}\right) ,$ we mean the order of $\infty$ as a zero of
$r$. That is $\circ\left( r_{\infty }\right)
=\mathrm{deg}(t)-\mathrm{deg}(s)$.

\end{enumerate}
\textbf{The four cases}\\

\noindent{\bf\large Case 1.} In this case $\left[ \sqrt{r}\right] _{c}$ and
$\left[ \sqrt{r}\right] _{\infty}$ stand for the Laurent series of
$\sqrt{r}$ at $c$ and the Laurent series of $\sqrt{r}$ at $\infty$
respectively. Furthermore, we define $\varepsilon(p)$ as follows:
if $p\in\Gamma,$ then $\varepsilon\left( p\right) \in\{+,-\}.$
Finally, the complex numbers $\alpha_{c}^{+},\alpha_{c}^{-},\alpha_{\infty}%
^{+},\alpha_{\infty}^{-}$ will be defined in the first step. If
the differential equation has no poles it only can fall in this
case.
\medskip

{\bf Step 1.} For each $c \in \Gamma'$ and for $\infty$ consider the
following possibilities:

\begin{description}

\item[$(c_{0})$] If $\circ\left(  r_{c}\right)  =0$, then
$$\left[ \sqrt {r}\right] _{c}=0,\quad\alpha_{c}^{\pm}=0.$$

\item[$(c_{1})$] If $\circ\left(  r_{c}\right)  =1$, then
$$\left[ \sqrt {r}\right] _{c}=0,\quad\alpha_{c}^{\pm}=1.$$

\item[$(c_{2})$] If $\circ\left(  r_{c}\right)  =2,$ and $$r= \cdots
+ b(x-c)^{-2}+\cdots,\quad \textrm{then}$$
$$\left[ \sqrt {r}\right]_{c}=0,\quad \alpha_{c}^{\pm}=\frac{1\pm\sqrt{1+4b}}{2}.$$

\item[$(c_{3})$] If $\circ\left(  r_{c}\right)  =2v\geq4$, and $$r=
(a\left( x-c\right)  ^{-v}+...+d\left( x-c\right)
^{-2})^{2}+b(x-c)^{-(v+1)}+\cdots,\quad \textrm{then}$$ $$\left[
\sqrt {r}\right] _{c}=a\left( x-c\right) ^{-v}+...+d\left(
x-c\right) ^{-2},\quad\alpha_{c}^{\pm}=\frac{1}{2}\left(
\pm\frac{b}{a}+v\right).$$

\item[$(\infty_{1})$] If $\circ\left(  r_{\infty}\right)  >2$, then
$$\left[\sqrt{r}\right]  _{\infty}=0,\quad\alpha_{\infty}^{+}=0,\quad\alpha_{\infty}^{-}=1.$$

\item[$(\infty_{2})$] If $\circ\left(  r_{\infty}\right)  =2,$ and
$r= \cdots + bx^{2}+\cdots$, then $$\left[
\sqrt{r}\right]  _{\infty}=0,\quad\alpha_{\infty}^{\pm}=\frac{1\pm\sqrt{1+4b}%
}{2}.$$

\item[$(\infty_{3})$] If $\circ\left(  r_{\infty}\right) =-2v\leq0$,
and
$$r=\left( ax^{v}+...+d\right)  ^{2}+ bx^{v-1}+\cdots,\quad \textrm{then}$$
$$\left[  \sqrt{r}\right]  _{\infty}=ax^{v}+...+d,\quad
and\quad \alpha_{\infty}^{\pm }=\frac{1}{2}\left(
\pm\frac{b}{a}-v\right).$$
\end{description}

{\bf Step 2.} Find $D\neq\emptyset$ defined by
$$D=\left\{
n\in\mathbb{Z}_{+}:n=\alpha_{\infty}^{\varepsilon
(\infty)}-%
{\displaystyle\sum\limits_{c\in\Gamma^{\prime}}}
\alpha_{c}^{\varepsilon(c)},\forall\left(  \varepsilon\left(
p\right) \right)  _{p\in\Gamma}\right\}  .$$ If $D=\emptyset$,
then we should start with the case 2. Now, if
$\mathrm{Card}(D)>0$, then for each $n\in D$ we search $\omega$
$\in\mathbb{C}(x)$ such that
$$\omega=\varepsilon\left(
\infty\right)  \left[  \sqrt{r}\right]  _{\infty}+%
{\displaystyle\sum\limits_{c\in\Gamma^{\prime}}}
\left(  \varepsilon\left(  c\right)  \left[  \sqrt{r}\right]  _{c}%
+{\alpha_{c}^{\varepsilon(c)}}{(x-c)^{-1}}\right).$$

{\bf Step 3}. For each $n\in D$, search for a monic polynomial
$P_n$ of degree $n$ with
\begin{equation}\label{recu1}
P_n'' + 2\omega P_n' + (\omega' + \omega^2 - r) P_n = 0.
\end{equation}
If success is achieved then $\zeta_1=P_n e^{\int\omega}$ is a
solution of the differential equation.  Otherwise, case 1 cannot hold.
\bigskip

\noindent{\bf\large Case 2.} \medskip

{\bf Step 1.} For each $c\in\Gamma^{\prime}$ and $\infty$ compute
non-empty sets $E_{c}\subset\mathbb{Z}$ and $E_{\infty}\subset\mathbb{Z}$ defined as follows:

\begin{description}
\item[($c_1$)] If $\circ\left(  r_{c}\right)=1$, then $E_{c}=\{4\}$.

\item[($c_2$)] If $\circ\left(  r_{c}\right)  =2,$ and $r= \cdots +
b(x-c)^{-2}+\cdots ,\ $ then $$E_{c}=\left\{
2+k\sqrt{1+4b}:k=0,\pm2\right\}.$$

\item[($c_3$)] If $\circ\left(  r_{c}\right)  =v>2$, then $E_{c}=\{v\}$.

\item[$(\infty_{1})$] If $\circ\left(  r_{\infty}\right)  >2$, then
$E_{\infty }=\{0,2,4\}$.

\item[$(\infty_{2})$] If $\circ\left(  r_{\infty}\right)  =2,$ and
$r= \cdots + bx^{2}+\cdots$, then $$E_{\infty }=\left\{
2+k\sqrt{1+4b}:k=0,\pm2\right\}.$$

\item[$(\infty_{3})$] If $\circ\left(  r_{\infty}\right)  =v<2$,
then $E_{\infty }=\{v\}$.
\end{description}

{\bf Step 2.} Find $D\neq\emptyset$ defined by
$$D=\left\{
n\in\mathbb{Z}_{+}:\quad n=\frac{1}{2}\left(  e_{\infty}-
{\displaystyle\sum\limits_{c\in\Gamma^{\prime}}} e_{c}\right)
,\forall e_{p}\in E_{p},\quad p\in\Gamma\right\}.$$ If
$D=\emptyset,$ then we should start the case 3. Now, if
$\mathrm{Card}(D)>0,$ then for each $n\in D$ we search a rational
function $\theta$ defined by
$$\theta=\frac{1}{2}
{\displaystyle\sum\limits_{c\in\Gamma^{\prime}}}
\frac{e_{c}}{x-c}.$$

{\bf Step 3.} For each $n\in D,$ search for a monic polynomial $P_n$
of degree $n$, such that
\begin{equation}\label{recu2}
P_n'''+3\theta
P_n''+(3\theta'+3\theta
^{2}-4r)P_n'+\left(
\theta''+3\theta\theta'
+\theta^{3}-4r\theta-2r'\right)P_n=0.
\end{equation}
 If $P_n$ does not
exist, then case 2 cannot hold. If such a polynomial is found, set
$\phi = \theta + P_n'/P_n$ and let $\omega$ be a solution
of
$$\omega^2 + \phi \omega + {1\over2}\left(\phi' + \phi^2 -2r\right)=
0.$$

Then $\zeta_1 = e^{\int\omega}$ is a solution of the differential
equation.
\bigskip

\noindent{\bf\large Case 3.} \medskip

{\bf Step 1.} For each $c\in\Gamma^{\prime}$ and $\infty$ compute
non-empty sets $E_{c}\subset\mathbb{Z}$ and $E_{\infty}\subset\mathbb{Z}$ defined as follows:

\begin{description}

\item[$(c_{1})$] If $\circ\left(  r_{c}\right)  =1$, then
$E_{c}=\{12\}$.

\item[$(c_{2})$] If $\circ\left(  r_{c}\right)  =2,$ and $r= \cdots +
b(x-c)^{-2}+\cdots$, then
\begin{displaymath}
E_{c}=\left\{ 6+k\sqrt{1+4b}:\quad
k=0,\pm1,\pm2,\pm3,\pm4,\pm5,\pm6\right\}.
\end{displaymath}

\item[$(\infty)$] If $\circ\left(  r_{\infty}\right)  =v\geq2,$ and $r=
\cdots + bx^{2}+\cdots$, then {\small $$E_{\infty }=\left\{
6+{12k\over m}\sqrt{1+4b}:\textrm{ }
k=0,\pm1,\pm2,\pm3,\pm4,\pm5,\pm6\right\},\textrm{ }
m\in\{4,6,12\}.$$}
\end{description}

{\bf Step 2.} Find $D\neq\emptyset$ defined by
$$D=\left\{
n\in\mathbb{Z}_{+}:\quad n=\frac{m}{12}\left(
e_{\infty}-{\displaystyle\sum\limits_{c\in\Gamma^{\prime}}}
e_{c}\right)  ,\forall e_{p}\in E_{p},\quad p\in\Gamma\right\}.$$
In this case we start with $m=4$ to obtain the solution,
afterwards $m=6$ and finally $m=12$. If $D=\emptyset$, then the
differential equation is not integrable because it falls in
case 4. Now, if $\mathrm{Card}(D)>0,$ then for each $n\in D$ with
its respective $m$, search for a rational function
$$\theta={m\over 12}{\displaystyle\sum\limits_{c\in\Gamma^{\prime}}}
\frac{e_{c}}{x-c}$$ and a polynomial $S$ defined as $$S=
{\displaystyle\prod\limits_{c\in\Gamma^{\prime}}} (x-c).$$

{\bf Step 3}. For each $n\in D$, with its respective $m$, search for a
monic polynomial $P_n=P$ of degree $n,$ such that $P$ can be
determined by the following polynomial recursion:
\begin{equation*}
\begin{aligned}
&P_{m}=-P,\\
&P_{i-1}=-SP_{i}'-\left( \left( m-i\right)
S'-S\theta\right)  P_{i}-\left( m-i\right)  \left(
i+1\right)  S^{2}rP_{i+1},\textrm{ for } i\in\{m,m-1,\ldots,1,0\},\\
&P_{-1}=0.
\end{aligned}
\end{equation*}
This can be done by using undetermined coefficients for $P$.
If $P$ does not exist, then the differential equation is not
integrable because it falls in case 4. Now, if $P$ exists search
$\omega$ such that $$ {\displaystyle\sum\limits_{i=0}^{m}}
\frac{S^{i}P}{\left( m-i\right)  !}\omega^{i}=0,$$ then a solution
of the differential equation is given by $$\zeta=e^{\int
\omega},$$ where $\omega$ is solution of the previous polynomial equation
of degree $m$.

\section*{Appendix B}
\renewcommand\theequation{B.\arabic{equation}}
In this appendix we show that the generating function for the case of linear death rates is given also by equation~\eqref{eq:genfun11}
if we assume $\delta<\beta$. Here we consider two separate cases:
\begin{itemize}

\item[(i)] $0\le z\le \frac{\delta}{\beta}$: here~\eqref{eq:integ11} reduces to
\begin{equation}
\ln G=\ln C+\frac{s}{\beta-\delta}\int \left(\frac{\beta}{\delta-\beta z}-\frac{1}{1-z}\right)dz,
\end{equation}
hence
\begin{equation}
G(z,s)=C(z)\left(\frac{1-z}{\delta-\beta z}\right)^{\frac{s}{\beta-\delta}},
\end{equation}
$C(z)$ being the constant obtained after integration. Variation of the constant in equation~\eqref{eq:ODE11} yields the first-order
ODE for $C(z)$,
\begin{equation}
(1-z)(\delta-\beta z)C'(z)\left(\frac{1-z}{\delta-\beta z}\right)^{\frac{s}{\beta-\delta}}=-z^{N_0}.
\end{equation}
We impose the condition $C(\delta/\beta)=0$ for $G(z,s)$ to be non-singular at $z=\frac{\delta}{\beta}<1$. Hence
\begin{equation}
C(z)=\int_z^{\delta/\beta} \frac{(\delta-\beta u)^{\frac{s}{\beta-\delta}-1}}{(1-u)^{\frac{s}{\beta-\delta}+1}}\,u^{N_0}du.
\end{equation}
Then the Laplace transform of the generating function can be written as
\begin{equation}
G(z,s)=\int_z^{\delta/\beta} \left(\frac{1-z}{\delta-\beta z}\right)\left[\left(\frac{1-z}{\delta-\beta z}\right)\left(\frac{\delta-\beta u}{1-u}\right)\right]^{\frac{s}{\beta-\delta}-1}\frac{u^{N_0}}{(1-u)^2}\,du.
\end{equation}
We change variable $u$ to $w(u):=\alpha\left(\frac{\delta-\beta u}{1-u}\right)$ with $\alpha:=\frac{1-z}{\delta-\beta z}$ and obtain
\begin{equation}
G(z,s)=\frac{1}{\beta-\delta}\int_0^1 w^{\frac{s}{\beta-\delta}-1}\left(\frac{w-\alpha\delta}{w-\alpha\beta}\right)^{N_0}dw.
\end{equation}
Finally we introduce a second change of variable, $w(t):=e^{-(\beta-\delta)t}$, which yields
\begin{equation}
G(z,s)=\int_0^{\infty} \left(\frac{w(t)-\alpha\delta}{w(t)-\alpha\beta}\right)^{N_0}e^{-st}dt,
\end{equation}
and the generating function is expressed as
\begin{equation}\label{eq:genfun2}
g(z,t)=\left(\frac{w(t)-\alpha\delta}{w(t)-\alpha\beta}\right)^{N_0}=\left[\frac{\delta-\beta z-(1-z)\delta e^{(\beta-\delta)t}}{\delta-\beta z-(1-z)\beta e^{(\beta-\delta)t}}\right]^{N_0},
\end{equation}
which exactly coincides with the expression obtained in Section~\ref{sec:Lap11}.

\item[(ii)] $\frac{\delta}{\beta}\le z \le 1$: in this case we can write
\begin{equation}
\ln G=\ln C-\frac{s}{\beta-\delta}\int \left(\frac{\beta}{\beta z-\delta}+\frac{1}{1-z}\right)dz,
\end{equation}
i.e.,
\begin{equation}
G(z,s)=C(z)\left(\frac{1-z}{\beta z-\delta}\right)^{\frac{s}{\beta-\delta}}.
\end{equation}
Variation of the constants implies
\begin{equation}
(1-z)(\beta z-\delta)C'(z)\left(\frac{1-z}{\beta z-\delta}\right)^{\frac{s}{\beta-\delta}}=z^{N_0},
\end{equation}
which can be integrated as
\begin{equation}
C(z)=\int_{\delta/\beta}^z \frac{(\beta u-\delta)^{\frac{s}{\beta-\delta}-1}}{(1-u)^{\frac{s}{\beta-\delta}+1}}\,u^{N_0}du.
\end{equation}
(notice the condition $C(\delta/\beta)=0$ for $G(z,s)$ to be finite at $z=\frac{\delta}{\beta}<1$). We can write
\begin{equation}
G(z,s)=\int_{\delta/\beta}^z \left(\frac{1-z}{\beta z-\delta}\right)\left[\left(\frac{1-z}{\beta z-\delta}\right)\left(\frac{\beta u-\delta}{1-u}\right)\right]^{\frac{s}{\beta-\delta}-1}\frac{u^{N_0}}{(1-u)^2}\,du.
\end{equation}
We change variables to $w(u):=\alpha\left(\frac{\beta u-\delta}{1-u}\right)$ with $\alpha:=\frac{1-z}{\beta z-\delta}$,
\begin{equation}
G(z,s)=\frac{1}{\beta-\delta}\int_0^1 w^{\frac{s}{\beta-\delta}-1}\left(\frac{w+\alpha\delta}{w+\alpha\beta}\right)^{N_0}dw,
\end{equation}
and after a second change of variable, $w(t):=e^{-(\beta-\delta)t}$, we finally obtain
\begin{equation}
G(z,s)=\int_0^{\infty} \left(\frac{w(t)+\alpha\delta}{w(t)+\alpha\beta}\right)^{N_0}e^{-st}dt.
\end{equation}
The generating function, in this case, is
\begin{equation}\label{eq:genfun3}
g(z,t)=\left(\frac{w(t)+\alpha\delta}{w(t)+\alpha\beta}\right)^{N_0}=\left[\frac{\beta z-\delta+(1-z)\delta e^{(\beta-\delta)t}}{\beta z-\delta+(1-z)\beta e^{(\beta-\delta)t}}\right]^{N_0},
\end{equation}
which coincides with~\eqref{eq:genfun11}.

\end{itemize}

\section*{Appendix C}
\renewcommand\theequation{C.\arabic{equation}}
In this Appendix we analyze in detail some cases of Kovacic's algorithm applied to equation~\eqref{eq:ODE11K2}. We are particularly
interested in those cases that impose a restriction on the values that $s$ can take to yield a closed-form solution, as well as in the cases where $m=N_0-1$ is a nonnegative integer. We start considering the combination of signs $(\varepsilon(\infty),\varepsilon(\delta/\beta),\varepsilon(1),\varepsilon(0))=(-,-,-,-)$. In this case,
$\hat{s}=\frac{s}{\delta-\beta}$ has to be a non-negative integer. Let $m=\hat{s}$ be a non-negative integer. Then we
form the rational function
\begin{equation}
\omega(z,s)=\frac{\alpha_0^-}{z}+\frac{\alpha_1^-}{z-1}+\frac{\alpha_{\delta/\beta}^-}{z-\delta/\beta}
=-\frac{N_0}{2z}+\frac{m-1}{2}\left(\frac{1}{1-z}+\frac{\beta}{\delta-\beta z}\right).
\end{equation}
Hence it can be checked that
\begin{equation}
\omega'(z,s)+\omega^2(z,s)-r(z,s)=\frac{m\beta[(N_0+m-1)z-N_0]}{z(1-z)(\delta-\beta z)}.
\end{equation}
We search for a polynomial of degree $m$ that satisfies $P_m''+2\omega P'_m+(\omega'+\omega^2-r)P_m=0$.
It turns out that $P_m(z)=(\delta-\beta z)^m$ satisfies the equation and there is a solution of the form
\begin{equation}
H(z,s)=(\delta-\beta z)^m\exp\left\{\int^z \omega(u,s)du\right\}=z^{-N_0/2}(1-z)^{\frac{1-m}{2}}(\delta-\beta z)^{\frac{1+m}{2}},
\end{equation}
where $m=\frac{s}{\delta-\beta}$ is a non-negative integer. Using~\eqref{eq:psi} we find the solution
\begin{equation}
G(z,s)=H(z,s)\psi(z,s)=\left(\frac{\delta-\beta z}{1-z}\right)^m,
\end{equation}
i.e., we obtain the solution given by equation~\eqref{eq:sol11} but specialized to non-negative integer values of the variable
$m=\frac{s}{\delta-\beta}$.

Similarly, consider the combination of signs $(\varepsilon(\infty),\varepsilon(\delta/\beta),\varepsilon(1),\varepsilon(0))=(-,+,+,-)$. Here
$-\hat{s}=-\frac{s}{\delta-\beta}$ has to be a non-negative integer. We define $m=-\hat{s}$ as a non-negative integer and
we find, as before,
\begin{equation}
\omega(z,s)=\frac{\alpha_0^-}{z}+\frac{\alpha_1^+}{z-1}+\frac{\alpha_{\delta/\beta}^+}{z-\delta/\beta}
=-\frac{N_0}{2z}+\frac{m-1}{2}\left(\frac{1}{1-z}+\frac{\beta}{\delta-\beta z}\right).
\end{equation}
Hence we obtain
\begin{equation}
\omega'(z,s)+\omega^2(z,s)-r(z,s)=\frac{m[\beta(N_0+m-1)z-\delta N_0]}{z(1-z)(\delta-\beta z)}.
\end{equation}
In this case the polynomial $P_m(z)=(1-z)^m$ satisfies~\eqref{eq:Pm} and there is a solution of the form
\begin{equation}
H(z,s)=(1-z)^m\exp\left\{\int^z \omega(u,s)du\right\}=z^{-N_0/2}(1-z)^{\frac{1+m}{2}}(\delta-\beta z)^{\frac{1-m}{2}}.
\end{equation}
Using~\eqref{eq:psi} we finally get
\begin{equation}
G(z,s)=H(z,s)\psi(z,s)=\left(\frac{\delta-\beta z}{1-z}\right)^{-m},
\end{equation}
which coincides with~\eqref{eq:sol11} because $m=-\frac{s}{\delta-\beta}$ in this case. Again, we recover the same
solution specialized to non-negative integer values of $m$.

Now we consider the combination of signs $(\varepsilon(\infty),\varepsilon(\delta/\beta),\varepsilon(1),\varepsilon(0))=(+,+,-,-)$. 
In this case $m=N_0-1$ has to be a non-negative integer; thus, we replace $N_0$ by $m+1$ in the following computations. Then we
form the rational function
\begin{equation}
\omega(z,s)=\frac{\alpha_0^-}{z}+\frac{\alpha_1^-}{z-1}+\frac{\alpha_{\delta/\beta}^+}{z-\delta/\beta}
=-\left(\frac{m+1}{2}\right)\frac{1}{z}+\left(\frac{s}{\delta-\beta}-1\right)\frac{1}{2(1-z)}-\left(\frac{s}{\delta-\beta}+1\right)\frac{\beta}{2(\delta-\beta z)}.
\end{equation}
Hence it can be checked that
\begin{equation}
\omega'(z,s)+\omega^2(z,s)-r(z,s)=-\frac{(m+1)s}{z(1-z)(\delta-\beta z)}.
\end{equation}
We search for a polynomial of degree $m$ that satisfies $P_m''+2\omega P'_m+(\omega'+\omega^2-r)P_m=0$ and we observe that for $m=0$, the polynomial $P_0=1$ does not satisfies such algebraic equation. Moreover, for $m>0$ the polynomial could exist with algebraic coefficients depending on $\beta$, $\delta$ and $s$. For example, the polynomial of degree $m=1$ is $P_1(z)=z-\delta/s$ only for $\delta=s-\beta$. 

Consider finally the combination of signs $(\varepsilon(\infty),\varepsilon(\delta/\beta),\varepsilon(1),\varepsilon(0))=(+,-,+,-)$. Then we
form the rational function
\begin{equation}
\omega(z,s)=\frac{\alpha_0^-}{z}+\frac{\alpha_1^+}{z-1}+\frac{\alpha_{\delta/\beta}^-}{z-\delta/\beta}
=-\left(\frac{m+1}{2}\right)\frac{1}{z}+\left(\frac{s}{\delta-\beta}+1\right)\frac{1}{2(1-z)}-\left(\frac{s}{\delta-\beta}-1\right)\frac{\beta}{2(\delta-\beta z)}.
\end{equation}
Then it holds
\begin{equation}
\omega'(z,s)+\omega^2(z,s)-r(z,s)=0.
\end{equation}
We search for a polynomial of degree $m$ that satisfies $P_m''+2\omega P'_m+(\omega'+\omega^2-r)P_m=0$ and we observe that for $m=0$, the 
polynomial $P_0=1$ satisfies such algebraic equation. In addition, for $m>0$ the polynomial exists with algebraic coefficients on $\beta$, $\delta$ and 
$s$. For example, we observe that the polynomial of degree $m=1$ is $P_1(z)=z$ when $\delta=0$ and $s=\beta-1$.

\noindent{\bf Remark}: although $N_0$ must be non-negative for its meaning in the birth-death processes, we can apply Kovacic algorithm when 
$m=-N_0-1$ is the degree of the polynomial $P_m$ assuming that $-N_0\in\mathbb{Z}^+$. Thus, we obtain the same results as above because the combinations
of signs $(\varepsilon(\infty),\varepsilon(\delta/\beta),\varepsilon(1),\varepsilon(0))=(+,+,-,-)$ and $(\varepsilon(\infty),\varepsilon(\delta/\beta),\varepsilon,\varepsilon(0))=(-,+,-,-)$ coincide with the combinations of signs 
$$(\varepsilon(\infty),\varepsilon(\delta/\beta),\varepsilon(1),\varepsilon(0))=(+,-,+,-)$$
and 
$$(\varepsilon(\infty),\varepsilon(\delta/\beta),\varepsilon(1),\varepsilon(0))=(-,-,+,-),$$ 
respectively.

\subsection*{Acknowledgements}

The authors kindly thank to the members of our Integrability Madrid Seminar for many fruitful discussions:  Rafael Hern\'andez-Heredero, Sonia Jim\'enez-Verdugo, Alvaro P\' erez-Raposo, Jos\'e Rojo-Montijano,  Sonia L. Rueda, Raquel S\'anchez-Cauce and Maria A. Zurro.

%\bibliographystyle{sigma}
%\bibliography{example}

\pdfbookmark[1]{References}{ref}
\LastPageEnding

\end{document}